\UseRawInputEncoding 
\documentclass[a4paper,11pt]{amsart}
\usepackage[T1]{fontenc}
\usepackage[english]{babel}
\usepackage{amsfonts, amsmath, amssymb, mathrsfs}
\usepackage{amsthm}
\usepackage{calc}
\newcommand*{\ditto}{--\texttt{"}--}
\usepackage{dsfont}
\usepackage{braket}
\usepackage{color}
\definecolor{darkbrown}{rgb}{.5,.2,0}

\usepackage{graphics}
\usepackage[colorlinks=true,urlcolor=darkbrown,citecolor=black,linkcolor=black]
{hyperref}
\def\doi[#1]{\href{https://doi.org/#1}{\texttt{DOI:\,#1}}}
\def\url[#1]{\href{#1}{\texttt{URL:\,#1}}}
\usepackage[noadjust]{cite}
\usepackage{soul} 
\usepackage{enumitem}
\setlist[enumerate]{%
topsep=.5ex plus0.5ex minus0.2ex,  
parsep=.5ex plus0.3ex minus0.1ex,  
itemsep=.5ex plus0.5ex minus0.2ex, 
leftmargin=4ex,    
itemindent=0ex,    
labelwidth=2ex,    
labelsep=1ex,      
listparindent=0ex, 
rightmargin=0ex,   
align=left,        
font=\rmfamily,
label=\arabic*)}
\theoremstyle{plain}
\newtheorem{Cor}{Corollary}
\newtheorem{Lem}{Lemma}

\newtheorem{Thm}{Theorem}
\theoremstyle{definition}
\newtheorem{Def}{Definition}

\newtheorem{Que}{Question}
\newtheorem{Rem}{Remark}
\newtheorem*{Cor*}{Corollary}
\newcommand{\C}{\mathbb{C}}
\newcommand{\N}{\mathbb{N}}
\newcommand{\R}{\mathbb{R}}

\newcommand{\cB}{\mathcal{B}}

\newcommand{\cH}{\mathcal{H}}

\newcommand{\cK}{\mathcal{K}}
\newcommand{\cP}{\mathcal{P}}

\newcommand{\fS}{\mathfrak{S}}

\newcommand{\id}{\mathds{1}}
\newcommand{\pro}{\mathrm{prod}}
\newcommand{\pure}{\mathrm{pure}}
\newcommand{\sep}{\mathrm{sep}}
\newcommand{\sym}{\mathrm{sym}}

\DeclareMathOperator\conv{conv}
\DeclareMathOperator\ext{ext}
\DeclareMathOperator\im{im}

\DeclareMathOperator\sgn{sgn}

\DeclareMathOperator\Tr{Tr}
\begin{document}
\title[Symmetric separable states and ground state energy]
{Decomposition of symmetric separable states and ground state energy of bosonic systems}
\author{Stephan Weis}
\begin{abstract}
We prove that every symmetric separable state admits a convex decomposition into 
symmetric pure product states. While the result is not new in itself, here we focus 
on convex geometry. 
We discuss the decomposition in the context of numerical ranges and ground state 
problems of infinite bosonic systems.
\end{abstract}
%
%
%
%
%
%
%
%
%
%
%
%
%
%
%
%
%
\maketitle
%
%
\section{Introduction}
Bosons are quantum mechanical particles described by density operators 
acting on the symmetric tensor product of a one-particle Hilbert space 
\cite{Messiah1999}. A classic result by St{\o}rmer \cite{Stormer1969} shows  
the $k$-particle marginal of an infinite bosonic system is a convex combination 
of symmetric pure product states. This result is known as the 
{\em quantum de Finetti theorem} in statistical mechanics 
\cite{Lewin-etal2014,Zeng-etal2019} and in quantum information theory 
\cite{Christandl-etal2009,LancienWinter2017}. Besides the de Finetti theorem,
the interrelation between symmetry and correlation is a more general topic in 
quantum information theory 
\cite{Christandl-etal2012,Augusiak-etal2012,QianChu2019}. 
\par
The mentioned symmetric pure product states are separable states, that is to say, 
they are convex combinations of product states. We prove the converse 
statement: every symmetric separable state admits a convex decomposition into 
symmetric pure product states. We are aware that the assertion is well known,
see \cite{Korbicz-etal2005,TothGuehne2009,Bohnet-Waldraff-etal2016} and 
\cite[Lemma 7.29]{Kraus2003}. Our proof stresses the convex geometry of the problem, 
namely that the set of symmetric states is a face of the set of all quantum states. 
The claim then follows because all one-particle marginals of a symmetric state are 
identical. Upon a sign change, our proof shows every antisymmetric state is entangled.
\par
The motivation behind this work is to examine linear images $\Pi_\sym$ of the set of 
symmetric pure product states. As per the quantum de Finetti theorem, the ground 
state energy of a local energy operator of an infinite bosonic system is the distance 
of the origin from a supporting hyperplane to $\Pi_\sym$. Studying the convex hull 
$\conv(\Pi_\sym)$ is therefore pivotal to understanding ground state phenomena of 
bosonic systems \cite{Zauner-etal2016,Chen-etal2017a,Chen-etal2017b}.
\par
Quantum information theory uses numerical ranges to represent linear images of 
subsets of quantum states 
\cite{Schulte-Herbrueggen-etal2008,Gawron-etal2010,Chen-etal2017b,Li-etal2020}. 
We show the set $\conv(\Pi_\sym)$ is a linear image of the set of symmetric separable 
states, a face of the set of separable states. This would enable studies of 
$\conv(\Pi_\sym)$ by taking advantage of methods used with the set of separable states 
\cite{AlfsenShultz2010,Kye2013,BengtssonZyczkowski2017,AubrunSzarek2017}.
\par
In Section~\ref{sec:dec-subspace} we recall that the set of quantum states supported 
on a subspace is a face of the set of all quantum states. 
Section~\ref{sec:dec-symm-sep} contains the main proofs and a question regarding the 
number of terms in the decomposition of a symmetric separable state. In Section~\ref{sec:jnr} we discuss numerical ranges, 
which we connect to the ground state energy of an infinite bosonic system in 
Section~\ref{sec:bosons}.
\par
%
%
\section{Decomposition of states supported on subspaces}
\label{sec:dec-subspace}
We recall aspects of the convex geometry of the set of quantum states and point
out that every state that contributes to a convex decomposition of a state 
supported on a subspace is itself supported on the same subspace. 
\par
Let $\cH$ be a finite-dimensional Hilbert space, $\cB(\cH)$ the Banach algebra of linear 
operators on $\cH$, and $\fS(\cH)$ the compact, convex set of {\em density operators} 
on $\cH$, that is to say, positive operators with trace one, also known as {\em states} 
or {\em quantum states}. If $\cH=\C^n$ for some $n\in\N$, we identify $\cB(\cH)$ with 
the set $M_n$ of $n$-by-$n$ matrices, and $\fS(\cH)$ with the compact, convex set of 
positive semi-definite matrices of trace one. 
\par
A {\em projector} on $\cH$ is an idempotent, self-adjoint operator $P\in\cB(\cH)$; 
this means that $P=P^2=P^\ast$. The set of projectors on $\cH$ is partially ordered 
by the relation $P\preceq Q$, which signifies that $Q-P$ is a positive operator. It 
is well known that the mapping that sends a projector $P$ on $\cH$ to the image 
$\im(P)=\{P(\varphi)\mid\varphi\in\cH\}$ is a lattice isomorphism from the set of 
projectors on $\cH$ to the set of subspaces of $\cH$ partially ordered by inclusion. 
\par
A subset $F$ of a convex set $C$ is a \emph{face} of $C$ if $F$ is convex and whenever 
$x\in F$ and $x=(1-\lambda)y+\lambda z$ for some $\lambda\in(0,1)$ and $y,z\in C$, then 
$y$ and $z$ are also in $F$. A point $x\in C$ is called an {\em extreme point} of $C$ 
if $\{x\}$ is a face of $C$. Let $\cK\subset\cH$ be a subspace. We say an operator 
$A\in\cB(\cH)$ is \emph{supported} on $\cK$ if the images 
$\im(A)=\{A(\varphi):\varphi\in\cH\}$ of $A$ and 
$\im(A^\ast)=\{A^\ast(\varphi):\varphi\in\cH\}$ of the adjoint $A^\ast$ are included 
in $\cK$. Let 
\begin{equation}\label{eq:supp}
\fS(\cH,\cK)
\end{equation}
denote the set of states on $\cH$ that are supported on $\cK$. It is well known that 
the mapping that sends a subspace $\cK$ of $\cH$ to the set $\fS(\cH,\cK)$ is a 
lattice isomorphism from the set of subspaces of $\cH$ to the set of faces of 
$\fS(\cH)$ partially ordered by inclusion \cite{AlfsenShultz2001}. All faces of the 
set of quantum states $\fS(\cH)$ are compact, convex sets.
\par
In particular, the set of extreme points of the set of quantum states $\fS(\cH)$, 
also called {\em pure states} on $\cH$, is the set of projectors of rank one, 
\begin{equation}\label{eq:pure}
\ext(\fS(\cH))
=\{\ket{\varphi}\!\!\bra{\varphi} \colon \varphi\in\cH, \braket{\varphi|\varphi}=1\}.
\end{equation}
In quantum information theory, the projector onto the span of a unit vector 
$\varphi\in\C^n$ is denoted by $\ket{\varphi}\!\!\bra{\varphi}$; both $\varphi$ and 
$\ket{\varphi}\!\!\bra{\varphi}$ are referred to as {\em pure states}. 
\par
Recall that the \emph{relative interior} of a convex subset $C$ of a Euclidean vector 
space is the interior of $C$ with respect to the topology of the affine hull of $C$.
\par
\begin{Lem}\label{lem:support}
Let $\cK\subset\cH$ be a subspace and let $\rho\in\fS(\cH,\cK)$ be a state supported 
on $\cK$. Let $\lambda_1,\dots,\lambda_d>0$ be strictly positive real numbers that 
sum up to one and let $\rho=\sum_{i=1}^d \lambda_i \rho_i$
be a convex combination of states $\rho_1,\dots,\rho_d\in\fS(\cH)$. Then each state 
$\rho_i$ is supported on $\cK$. If $\rho_i=\ket{\varphi_i}\!\!\bra{\varphi_i}$ is 
the pure state corresponding to a unit vector $\varphi_i\in\cH$, then $\varphi_i$ lies 
in $\cK$ for $i=1,\dots,d$.
\end{Lem}
\begin{proof}
By Theorem~6.9 of \cite{Rockafellar1970}, the state $\rho$ belongs to the relative 
interior of the convex hull $C=\conv(\{\rho_1,\dots,\rho_d\})$ of the points 
$\rho_1,\dots,\rho_d$. Since $\rho$ lies in the face $\fS(\cH,\cK)$, Theorem~18.1 of 
\cite{Rockafellar1970} shows that the whole convex set $C$ is included in 
$\fS(\cH,\cK)$. This proves the first statement; the second statement follows 
as the image of $\ket{\varphi_i}\!\!\bra{\varphi_i}$ is spanned by $\varphi_i$.
\end{proof}
The spectral projections of a state are supported on the same subspace as the state.
\par
\begin{Rem}\label{rem:spectral-proj}
Every state $\rho\in\fS(\cH,\cK)$ is the convex combination of at most $\dim_\C(\cK)$
pure states supported on $\cK$. It suffices to write $\rho$ as a convex combination 
of its spectral projections. Lemma~\ref{lem:support} shows the spectral projections 
are supported on $\cK$.
\end{Rem}
%
%
\section{Decomposition of symmetric separable states}
\label{sec:dec-symm-sep}
We show that every symmetric separable state is a convex combination of symmetric 
pure product states. The same proof shows all antisymmetric states are
entangled. 
\par
We label the units of a many-particle system by a finite set $\nu$. Let $\cH$ be a 
finite-dimensional Hilbert space, the one-particle Hilbert space. The Hilbert space of 
the subsystem $\mu$ is the tensor product $\cH^{\otimes\mu}=\bigotimes_{i\in \mu}\cH$, 
the identity operator of which we denote by $\id_\mu$, for all subsets 
$\mu\subset\nu$. We replace $\nu$ with its cardinality $|\nu|$ if convenient, for 
example to write symmetric product states 
$\sigma^{|\nu|}=\sigma^\nu=\bigotimes_{i\in\nu}\sigma$ where $\sigma\in\fS(\cH)$ is a 
state on the one-particle Hilbert space.
\par
Writing the complement as $\overline{\mu}=\nu\setminus\mu$, we define 
the {\em partial trace} over the subsystem $\overline{\mu}$ as the linear map 
$\Tr_{\overline{\mu}}:\cB(\cH^{\otimes\nu})\to\cB(\cH^{\otimes\mu})$, where
$\Tr_{\overline{\mu}}(A)$ is characterized by the conditions
\[\textstyle
\Tr\big(\Tr_{\overline{\mu}}(A)B\big)=\Tr\big(A(B\otimes\id_{\overline{\mu}})\big),
\qquad 
B\in\cB(\cH^{\otimes\mu}),
\] 
for all operators $A\in\cB(\cH^{\otimes\nu})$. If $\rho$ is a state on 
$\cH^{\otimes\nu}$, then $\Tr_{\overline{\mu}}(\rho)$ is a state on 
$\cH^{\otimes\mu}$, called the {\em marginal} of $\rho$ in the subsystem $\mu$. 
\par
A state $\rho\in\fS(\cH^{\otimes\nu})$ is a {\em product state} (with respect to 
$\cH$ and $\nu$) if it can be written in the form $\rho=\bigotimes_{i\in\nu}\rho_i$, 
where $\rho_i\in\fS(\cH)$ for all $i\in\nu$. Using pure states \eqref{eq:pure} we 
define the set of pure product states by
\begin{equation}\label{eq:pure-prod}
\fS_{\pure,\pro}(\cH,\nu)
=\Big\{\bigotimes_{i\in\nu}\rho_i \mid 
\rho_i\in\ext(\fS(\cH))\;\forall i\in\nu\Big\}. 
\end{equation}
A state $\rho\in\fS(\cH^{\otimes\nu})$ is {\em separable} if $\rho$ is a convex 
combination of product states. Otherwise, $\rho$ is {\em entangled}. We denote the 
set of separable states by 
\begin{equation}\label{eq:separable}
\fS_\sep(\cH,\nu).
\end{equation}
The possibility to decompose each factor of a product state into pure states and 
the distributive law show $\fS_\sep(\cH,\nu)$ is the convex hull of 
$\fS_{\pure,\pro}(\cH,\nu)$. In fact, $\fS_{\pure,\pro}(\cH,\nu)$ is the set of 
extreme points of $\fS_\sep(\cH,\nu)$ as every pure product state is an extreme 
point of $\fS(\cH^{\otimes\nu})$ and hence of $\fS_\sep(\cH,\nu)$. 
\par
The symmetric group $S_\nu$ of the finite set $\nu$ acts by linear automorphisms on 
the many-particle Hilbert space $\cH^{\otimes\nu}$ as per the law 
\[
\sigma\Big(\bigotimes_{i\in\nu}\varphi_i\Big)
=\bigotimes_{i\in\nu}\varphi_{\sigma^{-1}(i)},
\qquad\varphi_i\in\cH\;\forall i\in\nu,
\]
for all permutations $\sigma\in S_\nu$. The $\nu$-fold \emph{symmetric tensor product} 
of the Hilbert space $\cH$ is defined by
\begin{equation}\label{eq:symmetric}
\cH_\sym^{\otimes\nu}
=\Big\{\varphi\in\cH^{\otimes\nu} \mid 
\sigma(\varphi)=\varphi\;\forall\sigma\in S_\nu\Big\}.
\end{equation}
\par
It is easy to write down a basis for the symmetric tensor product.
\par
\begin{Rem}\label{rem:dicke-states}
Let $\ket{0},\ket{1},\dots,\ket{n-1}$ be an orthonormal basis of $\cH$. 
Let $T=(T_0,\dots,T_{n-1})$ be an $n$-tuple of non-negative integers that sum up 
to the cardinality $|\nu|$ of $\nu$ and define  
\[
\varphi_T
=c_T \sum_{\sigma\in S_\nu}\sigma\Big(\ket{
\underbrace{0\dots 0}_{\text{$T_0$}}
\underbrace{1\dots 1}_{\text{$T_1$}}
\dots
\underbrace{(n-1)\dots(n-1)}_{\text{$T_{n-1}$}}
}\Big),
\]
where $c_T=(|\nu|!\prod_{i=0}^{n-1}T_i!)^{-1/2}$ is a normalization constant. 
The orthonormal set $\{\varphi_T\}$ is a basis of the symmetric tensor product 
$\cH_\sym^{\otimes\nu}$, because $\{\varphi_T\}$ spans the image of the 
projector $\cH^{\otimes\nu}\to\cH^{\otimes\nu}$ defined by 
$\tfrac{1}{|\nu|!}\sum_{\sigma\in S_\nu}\sigma$, see \cite[Section~4.10]{Greub1978}. 
For example, the set of {\em Dicke states} \cite{MeillMeyer2017} with $k$ excitations, 
\[
\varphi_{(|\nu|-k,k)}
=(|\nu|!(|\nu|-k)!k!)^{-1/2}
\sum_{\sigma\in S_\nu}\sigma\Big(\ket{
\underbrace{0\dots 0}_{\text{$|\nu|-k$}}
\underbrace{1\dots 1}_{\text{$k$}}
}\Big),
\quad
k=0,\dots,|\nu|,
\]
is an orthonormal basis of the $\nu$-fold symmetric tensor product of $\C^2$. 
\end{Rem}
\begin{Def}[Symmetric States]
A state $\rho\in\fS(\cH^{\otimes\nu})$ is {\em symmetric} (with respect to $\cH$ and 
$\nu$) if $\rho$ is supported on the symmetric tensor product $\cH_\sym^{\otimes\nu}$ 
introduced in equation \eqref{eq:symmetric}. Using the notation from 
equation~\eqref{eq:supp} and~\eqref{eq:separable}, we denote the set of symmetric 
separable states by
\begin{equation}\label{eq:sym-sep}
\fS_{\sym,\sep}(\cH,\nu)
=\fS(\cH^{\otimes\nu},\cH_\sym^{\otimes\nu})\cap\fS_\sep(\cH,\nu).
\end{equation}
\end{Def}
Every convex combination of symmetric pure product states $\sigma^{\otimes\nu}$ is a 
symmetric separable state. Here we prove the converse statement that every symmetric 
separable state is a convex combination of symmetric pure product states.
\par
\begin{Thm}\label{thm:symmetric}
The set of symmetric separable states is the convex hull of the set of symmetric
pure product states,
\begin{equation}\label{eq:symm-conv}
\fS_{\sym,\sep}(\cH,\nu)
=\conv\big(\{\sigma^{\otimes\nu}\mid \sigma\in\ext(\fS(\cH))\}\big).
\end{equation}
The symmetric pure product states are the extreme points of $\fS_{\sym,\sep}(\cH,\nu)$,
\begin{equation}\label{eq:symm-ext}
\ext\big(\fS_{\sym,\sep}(\cH,\nu)\big)
=\big\{\sigma^{\otimes\nu}\mid \sigma\in\ext(\fS(\cH))\big\}.
\end{equation}
\end{Thm}
\begin{proof}
In order to prove \eqref{eq:symm-conv} it suffices to show that every symmetric 
separable state $\rho$ is a convex combination of symmetric pure product states. 
Being a separable state, $\rho$ admits a convex decomposition 
$\rho=\sum_{i=1}^d \lambda_i\ket{\varphi_i}\!\!\bra{\varphi_i}$ with positive 
coefficients $\lambda_1,\dots,\lambda_d>0$ into pure product states 
\[\textstyle
\varphi_i=\bigotimes_{j\in\nu}\psi_{i,j}\in\cH^{\otimes\nu}, 
\]
where $\psi_{i,j}\in\cH$ is a pure state for all $i=1,\ldots,d$ and $j\in\nu$. Since 
$\rho$ is supported on the symmetric tensor product $\cH_\sym^{\otimes\nu}$, 
Lemma~\ref{lem:support} shows $\varphi_i$ lies in the symmetric tensor product 
$\cH_\sym^{\otimes\nu}$ for all $i=1,\ldots,d$. This implies the one-particle 
marginals of $\varphi_i$ are identical to a pure state $\psi_i\in\cH$,
\[
\ket{\psi_i}\!\!\bra{\psi_i}
=\ket{\psi_{i,j}}\!\!\bra{\psi_{i,j}}
=\Tr_{\overline{\{j\}}}(\ket{\varphi_i}\!\!\bra{\varphi_i}),
\qquad j\in\nu,
\]
for all $i=1,\ldots,d$, which shows 
\[\textstyle
\rho
=\sum_{i=1}^d \lambda_i\ket{\varphi_i}\!\!\bra{\varphi_i}
=\sum_{i=1}^d \lambda_i(\ket{\psi_i}\!\!\bra{\psi_i})^{\otimes\nu}
\]
and proves the claim.
\par
The inclusion ``$\subset$'' in equation~\eqref{eq:symm-conv} implies 
that every extreme point of $\fS_{\sym,\sep}(\cH,\nu)$ must be a symmetric pure 
product state. Conversely, every symmetric pure product state is an extreme point 
of $\fS(\cH^{\otimes\nu})$ and hence of $\fS_{\sym,\sep}(\cH,\nu)$, which yields 
equation \eqref{eq:symm-ext}. 
\end{proof}
\begin{Que}
How many symmetric pure product states are required in the decomposition of a 
symmetric separable state $\rho$? The rank of $\rho$ may not be an upper bound on 
$d$ (as in Remark~\ref{rem:spectral-proj}) because the spectral projections of 
$\rho$ may be entangled. Examples of separable two-qubit states of rank three that 
require four pure product states in the decomposition are described in 
\cite{Wu2004}. How is the situation in the symmetric setting?
\end{Que}
\begin{Rem}
A similar decomposition as in Theorem~\ref{thm:symmetric} exists for the class of 
completely symmetric states, which is important in entanglement theory 
\cite{Chen-etal2019,QianChu2019}. These states are convex combination of real 
symmetric pure product states.
\end{Rem}
The $\nu$-fold \emph{antisymmetric tensor product} of the Hilbert space $\cH$ is 
defined by
\[
\Big\{\varphi\in\cH^{\otimes\nu} \mid 
\sigma(\varphi)=\sgn(\sigma)\varphi\;\forall\sigma\in S_\nu\Big\},
\]
where $\sgn(\sigma)$ is the sign of a permutation $\sigma\in S_\nu$, which is $+1$ 
if $\sigma$ is a composition of an even number of transpositions, and $-1$ else. A 
state on $\cH^{\otimes\nu}$ is \emph{antisymmetric} if it is supported on the 
antisymmetric tensor product.
\par
\begin{Thm}\label{thm:antisymmetric}
If $\nu$ has at least two elements, then every antisymmetric state is entangled.
\end{Thm}
\begin{proof}
The same reasoning as in Theorem~\ref{thm:symmetric} shows every antisymmetric 
separable state is a convex combination of symmetric pure product states. Hence,
the existence of an antisymmetric separable state would lead to a contradiction.
\end{proof}
%
%
\section{Joint Numerical Ranges}
\label{sec:jnr}
Numerical ranges, originally a topic of matrix theory, have become indispensable 
tools in quantum information theory
\cite{Schulte-Herbrueggen-etal2008,Gawron-etal2010,Chen-etal2017b,Li-etal2020}
to study linear images of certain subsets of quantum states.
We show two notions of symmetric numerical ranges to be equal.
\par
Let $A_1,\dots,A_m\in \cB(\cH^{\otimes\nu})$ be a sequence of hermitian operators.
Using the sets of pure states~\eqref{eq:pure}, pure product 
states~\eqref{eq:pure-prod}, and separable states~\eqref{eq:separable}, 
we define the numerical ranges \cite{Gawron-etal2010,Chen-etal2017b} 
\begin{align*}
\Lambda
&=\left\{\Tr(\rho A_i)_{i=1}^m \mid \rho\in\ext(\fS(\cH^{\otimes\nu}))\right\},
&& \text{\emph{(joint numerical range)}}\\
\Pi
&=\left\{\Tr(\rho A_i)_{i=1}^m \mid \rho\in\fS_{\pure,\pro}(\cH,\nu)\right\},
&& \text{\emph{(\parbox{\widthof{joint}}{\centering \ditto} product numerical range)}}\\
\Theta
&=\left\{\Tr(\rho A_i)_{i=1}^m \mid \rho\in\fS_\sep(\cH,\nu)\right\}.
&& \text{\emph{(\parbox{\widthof{joint}}{\centering \ditto} separable 
\parbox{\widthof{blabla}}{\centering \ditto}\hspace{1em}
\parbox{\widthof{range}}{\centering \ditto})}}
\end{align*}
Symmetric numerical ranges have been studied, too \cite{Chen-etal2017b}. Using 
symmetric pure product states, we define the set
\begin{equation}\label{eq:jssnr}
\Pi_\sym
=\left\{\Tr(\sigma^{\otimes\nu}A_i)_{i=1}^m \mid \sigma\in\ext(\fS(\cH))\right\},
\end{equation}
which we call \emph{joint symmetric product numerical range}. Employing the set of 
symmetric separable states \eqref{eq:sym-sep}, we define the set
\[
\Theta_\sym
=\left\{\Tr(\rho A_i)_{i=1}^m \mid \rho\in \fS_{\sym,\sep}(\cH,\nu) \right\},
\]
which we call \emph{joint symmetric separable numerical range}. 
\par
Recalling equation \eqref{eq:pure-prod}, we observe the set of separable states is 
the convex hull of the set of pure product states; therefore $\Theta=\conv(\Pi)$. 
The analogue is true for symmetric numerical ranges.
\par
\begin{Cor}\label{cor:jssnr}
The joint symmetric separable numerical range $\Theta_\sym$ is the convex hull of the 
joint symmetric product numerical range $\Pi_\sym$.
\end{Cor}
\begin{proof}
Theorem~\ref{thm:symmetric} proves the set of symmetric separable states 
$\fS_{\sym,\sep}(\cH,\nu)$ is the convex hull of the set 
$\{\sigma^{\otimes\nu}\mid \sigma\in\ext(\fS(\cH))\}$ of symmetric pure product states. 
The claim follows as $\Theta_\sym$ is the image of the former set and $\Pi_\sym$ is the 
image of the latter set under the same linear map.
\end{proof}
Corollary~\ref{cor:jssnr} clarifies a relationship between two numerical ranges.
\par
\begin{Rem}
Chen et.~al \cite{Chen-etal2017b} have defined the symmetric version of the joint 
separable numerical range $\Theta$ to be the convex hull $\conv(\Pi_\sym)$ of the 
joint symmetric product numerical range $\Pi_\sym$. This brings to mind the 
equation $\Theta=\conv(\Pi)$. The analogous equation $\Theta_\sym=\conv(\Pi_\sym)$ 
shows the set $\conv(\Pi_\sym)$ is a linear image of the set of symmetric 
separable states $\fS_{\sym,\sep}(\cH,\nu)$, thereby corroborating $\conv(\Pi_\sym)$ 
is a natural choice for a symmetric version of $\Theta$.
\end{Rem}
Corollary~\ref{cor:jssnr} offers a new approach to study the set 
$\conv(\Pi_\sym)$. 
\par
\begin{Rem}
The identity $\conv(\Pi_\sym)=\Theta_\sym$ allows us to study the set
$\conv(\Pi_\sym)$ \emph{via} the set of symmetric separable states 
$\fS_{\sym,\sep}(\cH,\nu)$, the preimage of the numerical range 
$\Theta_\sym$. The convex set $\fS_{\sym,\sep}(\cH,\nu)$ is a face of the set of 
separable states $\fS_\sep(\cH,\nu)$, for which algebraic tools are available
\cite{AlfsenShultz2010,Kye2013,BengtssonZyczkowski2017,AubrunSzarek2017}. 
Studying linear images of $\fS_{\sym,\sep}(\cH,\nu)$ in the context of statistical
mechanics is interesting already for two qubits \cite{Chen-etal2017a,Chen-etal2017b}, 
where every separable state has a positive partial transpose \cite{Horodecki-etal1996}. 
\end{Rem}
%
%
\section{The Quantum de Finetti Theorem}
\label{sec:bosons}
This section highlights the role the joint symmetric separable numerical range 
$\Theta_\sym$ plays in the statistical mechanics of bosonic systems.
\par
Symmetric states on the many-particle Hilbert space $\cH^{\otimes\nu}$ are 
called {\em bosonic states} in physics 
\cite{Messiah1999,Lewin-etal2014,Chen-etal2017b}. Due to the symmetry, the marginals 
of a symmetric state $\rho\in\fS(\cH^{\otimes\nu},\cH_\sym^{\otimes\nu})$ are 
themselves symmetric states. The marginals of $\rho$ depend on the subsystems only 
through their sizes,
\[
\Tr_{\overline{\mu}}(\rho)
=\Tr_{\overline{\eta}}(\rho)
\qquad
\text{for all $\mu,\eta\subset\nu, |\mu|=|\eta|$.}
\]
Thus, one studies the set of {\em $k$-particle $N$-representable density operators}
\[
\cP_N^{(k)}
=\left\{\Tr_{\overline{\{1,\ldots,k\}}}(\rho) \mid 
\rho\in\fS(\cH^{\otimes N},\cH_\sym^{\otimes N})\right\}
\subset\fS(\cH^{\otimes k},\cH_\sym^{\otimes k})
\]
with reference to the total number $N=|\nu|$ of units and the size $k$ of the 
subsystem. According to a classic result by St{\o}rmer \cite{Stormer1969}, 
also known as the {\em quantum de Finetti theorem} 
\cite{Lewin-etal2014,Chen-etal2017b,Zeng-etal2019}, we have 
\begin{equation*}
\bigcap_{N\geq k}\cP_N^{(k)}
=\conv\big(\{\sigma^{\otimes k}\mid \sigma\in\ext(\fS(\cH))\}\big),
\qquad 
k\in\N.
\end{equation*}
In other words, the set of all possible $k$-particle marginals of an infinite 
bosonic system is the convex hull of the set of symmetric pure product states. 
By Theorem~\ref{thm:symmetric}, 
\begin{equation}\label{eq:stormer}
\bigcap_{N\geq k}\cP_N^{(k)}
=\fS_{\sym,\sep}(\cH,k),
\qquad 
k\in\N,
\end{equation}
where $\fS_{\sym,\sep}(\cH,k)=\fS_{\sym,\sep}(\cH,\nu)$ is the set of symmetric 
separable states on a system $\nu$ of size $|\nu|=k$.
\par
Quantum phase transitions are associated with abrupt changes of the ground state 
energy of an energy operator
\[ 
H(x)=x_1H_1+x_2H_2+\dots+x_mH_m
\] 
on $\cH^{\otimes N}$ under smooth changes of the parameter $x=(x_1,\dots,x_k)\in\R^m$ 
for large system sizes $N\to\infty$. In many cases the energy operators $H_1,\dots,H_m$ 
are {\em $k$-local Hamiltonians} for some fixed $k\in\N$. This means 
\[\textstyle
H_i=\sum_{\mu\subset\nu,|\mu|=k}H_{i,\mu}\otimes\id_{\overline{\mu}},
\qquad i=1,\dots,m,
\]
where $H_{i,\mu}\in\cB(\cH^{\otimes\mu})$ interacts with at most $|\mu|=k$ units
of the total system $\nu=\{1,\ldots,N\}$. The ground state energy of a bosonic system 
with energy operator $H(x)$ is  
\[\textstyle
\min_{\rho\in\fS(\cH^{\otimes N},\cH_\sym^{\otimes N})}
\Tr(H(x)\rho)
=\min_{\sigma\in\cP_N^{(k)}}
\Tr(A(x)\sigma),
\]
where $A(x)=\sum_{i=1}^m x_iA_i$ and 
\[\textstyle
A_i=\sum_{\mu\subset\nu,|\mu|=k}H_{i,\mu}\in\cB(\cH^{\otimes k}),
\qquad i=1,\dots,m.
\]
For this sum to be well-defined, one may assume each operator $H_{i,\mu}$ is supported 
on the symmetric tensor product $\cH_\sym^{\otimes\mu}$ or replace $H_{i,\mu}$ with 
its compression onto $\cH_\sym^{\otimes\mu}$. Provided the convergence of the energy 
operators $A_i$ is guaranteed in the thermodynamic limit $N\to\infty$, St{\o}rmer's 
result in the version of equation \eqref{eq:stormer} shows the ground state energy of 
$H(x)$ is
\[
\min_{\sigma\in\fS_{\sym,\sep}(\cH,k)}\Tr(A(x)\sigma)
=\min_{y\in\Theta_\sym}x_1y_1+x_2y_2+\dots+x_my_m,
\]
where $\Theta_\sym$ is the joint symmetric separable numerical range \eqref{eq:jssnr}
with respect to the operators $A_1,\dots,A_m\in\cB(\cH^{\otimes k})$. 
\par
In other words, the ground state energy of $H(x)$ is the distance of the origin from 
the supporting hyperplane to $\Theta_\sym$ with inner normal vector $x$. This rationale 
strongly motivates to study the geometry of the numerical range $\Theta_\sym$. The 
key-feature of a ruled surface on the boundary of $\Theta_\sym$ is an expression of a 
phase transition \cite{Zauner-etal2016,Chen-etal2017a,Chen-etal2017b} and deserves a 
thorough investigation.
\par
%
%
%
\vspace{\baselineskip}
\noindent
{\footnotesize
Acknowledgements.
I thank Karol {\.Z}yczkowski for discussions on numerical ranges and their 
use in quantum mechanics. I thank Lin Chen, who kindly pointed out to me
the difference between symmetric and completely symmetric states. 
I thank an anonymous referee, who made me aware of the literature  
\cite{Korbicz-etal2005,Kraus2003,TothGuehne2009,Bohnet-Waldraff-etal2016}.}
%
%
\bibliographystyle{plain}

%
%
%
\vspace{\baselineskip}
\parbox{10cm}{%
Stephan Weis\\
Berlin, Germany\\
e-mail \texttt{maths@weis-stephan.de}}
\vspace{\baselineskip}
\par\noindent
\end{document}